\documentclass[conference]{IEEEtran}

\usepackage[latin1]{inputenc}
\usepackage{graphicx}
\usepackage{amssymb,amsmath,amsfonts,amsthm}
\usepackage{algorithm,algorithmic}
\usepackage{cite}
\usepackage{float}
\usepackage{graphics}
\usepackage{subfigure}
\usepackage{xspace}
\usepackage[usenames,dvipsnames]{color}
\usepackage{epsfig}

\newtheorem{proposition}{Proposition}
\newtheorem{remark}{Remark}

\begin{document}

\title{Information-Energy Capacity Region for SWIPT Systems with Power Amplifier Nonlinearity}

\author{\authorblockN{Ioannis Krikidis}
\authorblockA{Department of Electrical and Computer Engineering, \\ University of Cyprus, Cyprus}
\authorblockA{E-mail:  krikidis@ucy.ac.cy}}
\maketitle

\begin{abstract}
We study the information-energy capacity region (IE-CR) of an additive white Gaussian noise (AWGN) channel in the presence of  high-power amplifier (HPA) nonlinearity. Specifically, we consider a three-node network consisting of one transmitter, one information receiver and one energy receiver and we study the capacity-achieving input distribution under i) average-power, peak-power constraints at the transmitter, b) HPA nonlinearity at the transmitter, and c) nonlinearity of the energy harvesting circuit at the energy receiver. We prove that the input distribution is discrete and finite and we derive closed form expressions for the special cases of maximizing the harvested energy and maximizing the information capacity. We show that HPA significantly reduces the achievable IE-CR; to compensate HPA nonlinearity, a predistortion technique is also discussed and evaluated in terms of IE-CR performance.
\end{abstract}

\section{Introduction}

Simultaneous wireless information and power transfer (SWIPT) is a new technology, where a dedicated radio-frequency (RF) transmitter conveys information and energy to wireless devices by using the same radio waveform \cite{CLE}.  It is a promising technology for future communication networks, which are characterized by a massive number of low-power devices (e.g., Internet of Things). The key idea of SWIPT has been proposed by Varshney in \cite{VAR}, where the fundamental trade-off between information and energy transfer has been introduced for a simple point-to-point channel; this work has been extended in \cite{GRO} for a parallel-links point-to-point channel. More recent works study the integration of SWIPT in more complex network configurations e.g., multiple-antenna systems \cite{RUI2}, multiple-access networks \cite{KRI}, multiple-antenna cellular networks \cite{LAM}, etc.     

One of the main particularities of a SWIPT network is that the wireless power transfer channel is highly nonlinear (in contract to the linear information transfer channel). Recent studies take into account the nonlinearity of the rectification circuit, and study the impact of waveform design and/or input distribution on the achieved information-energy capacity region (IE-CR). For instance, the work in \cite{BRU} models the rectifier's behaviour and introduces a mathematical framework to design waveforms that exploit nonlinearity. On the other hand, the authors in \cite{VAR2} study the input distribution that maximizes IE-CR for an additive white Gaussian noise (AWGN) channel under statistical constraints (first/second moment statistics) on the input distribution. By relaxing these constraints, the authors in \cite{MOR} study the input distribution under general average-power (AP) and a peak-power (PP) constraints at the transmitter. By using the mathematical framework in \cite{SMI}, they prove that the input distribution is unique, discrete with a finite number of mass points.     

According to experimental studies, signals with high peak-to-average-power-ratio (PAPR) e.g., multi-sine signals, increase the direct-current (DC) output power of the rectifier and enhance the IE-CR performance \cite{CLE,BRU,KIM}, in comparison to constant-envelope signals. However, signals with high PAPR are more sensitive to high-power amplifier (HPA) nonlinearities that can significantly degrade the quality of the communication \cite{QI}. Despite this fundamental experimentally-validated observation, previous works do not take into account the effects of HPA on the achieved IE-CR and assume that the RF power amplifier operates always in the linear regime.  

In this paper, we study the fundamental limits of a SWIPT system which is characterized by HPA nonlinearities at the transmitter. By taking into account a memoryless HPA model i.e., solid state power amplifier (SSPA) model \cite{QI}, we characterize the IE-CR for a three-node real-valued AWGN channel, under both AP and PP constraints at the transmitter as well as rectification nonlinearities at the energy receiver. We study the input distribution that maximizes IE-CR by formulating appropriate convex optimization problems over the input distribution and we provide simplified closed-form expressions, when the design maximizes information/energy transfer. We show that HPA significantly reduces the IE-CR, while a non-tradeoff between information and energy is observed for low PP constraints. Finally, a deterministic digital predistortion (PD) that inverses the HPA nonlinearities and linearizes the below-saturation regime is discussed; we show that PD enhances the IE-CR performance when HPA nonlinearity is more severe.   \\  
\noindent {\it Notation:}  Lower case bold symbols denote vectors, $\mathbb{E}[\cdot]$ represents the expectation operator, $\succeq$ denotes componentwise inequality, the superscript $(\cdot)^{\top}$ denotes 
transpose, and $\mathbb{P}(X)$ is the probability of the event $X$; $f(x)\nearrow (x_1, x_2)$ and $f(x)\searrow (x_1, x_2)$ denote that the function $f(x)$ is monotonically increasing/decreasing in the interval $(x_1, x_2)$, respectively.

\begin{figure}
\centering\
\includegraphics[width=0.8\linewidth]{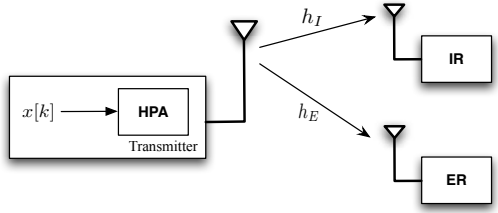}
\vspace{-0.5cm}
\caption{System model consisting of one transmitter, an information receiver and an energy receiver.} \label{model}
\end{figure}

\section{System model}

We assume a simple SWIPT topology consisting on one transmitter, one information receiver (IR), and one energy receiver (ER) \cite{MOR}. All the terminals are equipped with a single antenna; the IR converts the received signal to the baseband to decode the transmit information, while the ER harvests energy from the received RF signal. The transmitter transmits a pulse-amplitude modulated signal $x(t)=\sum_{k=-\infty}^{\infty}x[k]p(t-kT)$ with an average power $\sigma_x^2$, where $p(t)$ is the rectangular pulse shaping filter (i.e., $p(t)=1$ for $0<t\leq T$), $T$ is the symbol interval, and $x[k]$ is the information symbol at time index $k$, modeled as the realization of an independent and identically distributed (i.i.d) real random variable $X$ with cumulative distribution function $F_X(x)$. We assume a normalized symbol interval $T=1$ and thus the measures of energy and power become identical and therefore are used equivalently. Fig. \ref{model} schematically presents the system model.  

The modulated signal is amplified by the HPA at the RF chain, which causes amplitude distortion and nonlinearity on the transmitted amplitude-modulated signal $x(t)$. Specifically, the output of the nonlinear HPA can be written as $\hat{x}[k]=d(x[k])$ (i.e., random variable $\hat{X}=d(X)$), where $d(\cdot)$ denote the AM-to-AM conversion which is given by the considered SSPA HPA model \cite{QI} i.e., 
\begin{align}
&d(r)=
\frac{r}{\left[1+\left(\frac{r}{A_{s}} \right)^{2\beta} \right]^{\frac{1}{2\beta}}}, 
\end{align}
where $A_{s}$ is the output saturation voltage, and $\beta$  represents the smoothness of the transition from the linear regime to the saturation. Let $A_0$ denote the minimum input voltage that drives the HPA output to the saturation i.e., $A_0=\min_{d(r)=A_s} r$ with $A_0=A_{s}$ for $\beta\gg 1$. Fig. \ref{HPA} presents the input-output voltage characteristics for the considered HPA model; as $\beta$ increases, the nonlinear transition regime (below saturation) of the HPA is linearized.        

We consider transmission over an AWGN channel with fixed and known channel fading \cite{VAR2,MOR}. The equivalent baseband received signal at the IR is given by 
\begin{align}
y(t)=h_{I}\hat{x}(t)+n(t),
\end{align} 
where $h_{I}\in\Re$ is the channel fading gain (constant) and $n(t)$ is the real-valued Gaussian noise with unit variance. We define the conditional probability 
\begin{align}
p(y|x)=\frac{1}{\sqrt{2\pi}}\exp \left[-\frac{(y-h_I d(x))^2}{2} \right].
\end{align}

The ER converts the received RF signal to DC power through a nonlinear rectification circuit\footnote{We assume that energy harvesting from background Gaussian noise is negligible and is ignored \cite{RUI2}.}. If $h_{E} \in \Re$ is the fading channel gain (constant) for the link transmitter-ER, the average harvested energy is captured by the following quantity (monotonically increasing with the average harvested energy) \cite{MOR,RUI}
\begin{align}
\mathcal{E}=&\mathbb{E}\big[I_0(B h_E |\hat{X}|) \big],
\end{align}
where $I_0(\cdot)$ denote the modified Bessel function of the first kind and order zero, and $B$ is a constant that depends on the characteristics of the rectification circuit\footnote{The considered energy harvested-based quantity results in $\mathcal{E}\geq 1$; $\mathcal{E}=1$ corresponds to a zero DC power delivered to the load \cite[Eq. (4)]{MOR}.}. 

In addition to the transmit AP constraint $\sigma_x^2$, the transmitter has a PP constraint to control the negative effects of saturation which characterizes both the transmitter and the ER due to HPA and the diode breakdown, respectively \cite{MOR}; the PP constraint can be expressed as $|X|\leq A$, where $A$ is the peak amplitude. 
     
\begin{figure}
\centering\
\includegraphics[width=0.85\linewidth]{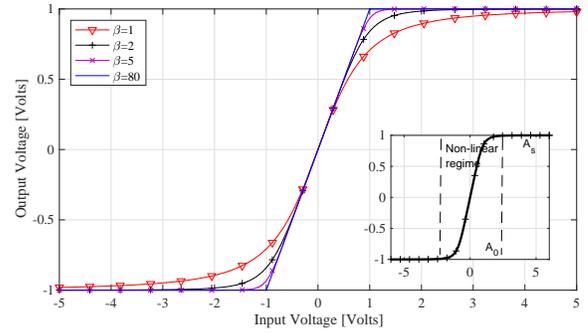}
\vspace{-0.3cm}
\caption{Input-output voltage characteristics for the SSPA power amplifier model; $A_{s}=1$, $\beta=\{1,2,5, 80\}$.} \label{HPA}
\end{figure}

\section{Information-energy capacity region}

We consider firstly the case where the IR is not present/active. In this case, we design the input distribution $F_X$ under both AP and PP constraints to maximize the power delivered at ER. The considered design problem can be formulated as
\begin{equation}
\begin{array}{rrclcl}
(P1)\;\;\displaystyle \max_{F_X} & \multicolumn{3}{l}{\mathbb{E}\big[I_0(B h_E |\hat{X}|)\big]}\\
\textrm{s.t.} & \mathbb{E}[X^2]\leq \sigma_x^2\\
&|X|\leq A.
\end{array}
\end{equation}
Although the problem (P1) is a linear optimization problem and can be solved with standard convex optimization tools (e.g., CVX), we can provide a closed form solution. The following proposition gives the solution to (P1) and the associated input distribution for the different cases.  
\begin{proposition}\label{pr1}
The maximum average harvested energy and the associated mass point distribution are given by
\begin{align}
\mathcal{E}_{\max}=p  I_0\big(B h_E d(\lambda) \big)+(1-p),\;\text{where}
\end{align}
\begin{itemize}
\item If $A^2\leq \sigma_x^2$, we have $p=1$, $\lambda=A$, and mass point distribution $\Pi_A=\frac{1}{2}(\delta_{-A}+\delta_{A})$. 
\item If  $\sigma_x^2\leq A^2$ and $g(x)\searrow (\sigma_x, A)$, we have $p=1$, $\lambda=\sigma_x$, and mass point distribution $\Pi_A=\frac{1}{2}(\delta_{-\sigma_x}+\delta_{\sigma_x})$. 
\item If $\sigma_x^2\leq A^2$ and $g(x)\nearrow (\sigma_x, A)$, we have $p=\frac{\sigma_x^2}{A^2}$, $\lambda=A$, and mass point distribution $\Pi_A=\frac{\sigma_x^2}{2A^2}(\delta_{-A}+\delta_{A})+\left( 1-\frac{\sigma_x^2}{A^2} \right)\delta_{0}$.
\item If $\sigma_x^2\leq A^2$ and the function $g(x)\nearrow (\sigma_x, A')$ and $g(x)\searrow (A',  A)$, we have $p=\frac{\sigma_x^2}{A'^2}$, $\lambda=A'$, and mass point distribution  $\Pi_A=\frac{\sigma_x^2}{2A'^2}(\delta_{-A'}+\delta_{A'})+ \left( 1-\frac{\sigma_x^2}{A'^2} \right)\delta_{0}$, with $A'\approx A_s$ for $\beta\gg1$,
\end{itemize}
where $\delta_x$ is the Dirac measure (point mass) concentrated at $x$, and $g(x)=\frac{1}{x^2} [I_0(B h_E d(x))-1]$.
\end{proposition} 
\begin{proof}
The proof is given in the Appendix.
\end{proof}
In case that the IR is active, the achieved information capacity with $\Pi_A$ is equal to $I_{\min}=\int_{-\infty}^{\infty}\sum_{j}p(y|x_j)p_j \log_2 \frac{p(y|x_j)}{\sum_{j} p(y|x_j)p_j} dy$ where $p_j=\mathbb{P}[X=x_j]$; since $\Pi_A=\sum_{j}p_j \delta_{x_j}$ is a binary/ternary distribution, the complexity of the numerical computation is low.  

Then, we consider the case where the target of the system is to maximize the Shannon information capacity under both AP and PP constraints. The problem can be formulated as follows
\begin{equation}
\begin{array}{rrclcl}
(P2)\;\;\displaystyle \max_{F_X} & \multicolumn{3}{l}{I(X;Y)}\\
\textrm{s.t.} & \mathbb{E}[X^2]\leq \sigma_x^2\\
&|X|\leq A,
\end{array}
\end{equation}
where $I(X;Y)$ is the average mutual information between the channel input $X$ and the channel output $Y=h_I \hat{X}+N$ with $\hat{X}=d(X)$. Given that the input probability distribution is constrained to $(-A, A)$, the mutual information is given by 
\begin{align}
I(X;Y)=\int_{-A}^{A} \int_{-\infty}^{\infty} p(y|x)\log_2 \left(\frac{p(y|x)}{p(y;F_X)} \right)dy dF_X,
\end{align}
where $p(y;F_X)$ is the marginal output probability density function given an input distribution $F_X$. Due to the considered PP constraint and the nonlinearity of the HPA model, we can show that the optimal input probability function $F_X$ is unique, finite and discrete. The proof requires the application of a systematic methodology and is similar to the analysis in \cite{MOR,SMI}. Given the finiteness/discretness of the input distribution, (P2) can be discretized and reformulated by the following convex optimization problem (corresponding to the capacity of a discrete memoryless channel (DMC) with AP/PP constraints) i.e.,
\begin{equation}
\begin{array}{rrclcl}
(P3)\;\;\displaystyle \max_{\pmb{p}} & \multicolumn{3}{l}{I\triangleq\sum_{i=1}^m \sum_{j=1}^n p_{ij} p_{j}\log_2 \frac{p_{ij}}{\sum_{k=1}^n p_{ik}p_{k}}}\\
\textrm{s.t.} & \mathbb{E}[X^2]\leq \sigma_x^2\\
&|X|\leq A \\
&\pmb{p}\succeq 0,\;\;\; \pmb{1}^\top \pmb{p}=1, 
\end{array}
\end{equation}
where $\pmb{1}$ denotes a vector with ones, $p_{ij}=\mathbb{P}(Y=y_i|X=x_j)$, and $\pmb{p}=[p_1, p_2,\ldots, p_n]^\top$. The above formulation discretizes the intervals $x \in (-A, A)$ and $y \in (-\Gamma, \Gamma)$ (where $\Gamma\gg A$) with sufficiently small step size $\Delta x\rightarrow 0, \Delta y \rightarrow 0$ to form the input ($n=2A/\Delta x$ mass points) and the output ($m=2\Gamma/\Delta y$ mass points) symbol set, respectively. Formulation (P3) is a convex optimization problem where the objective function is concave in $\pmb{p}$; therefore can be solved by using standard convex optimization tools (e.g., CVX). It is worth noting that (P3) can be also solved by using the Blahut-Arimoto algorithm for constrained discrete channels, which numerically computes the capacity of DMC with constraints in the input distribution \cite{BLA}. If $\pmb{p}^*$ is the solution to (P3), the maximum mutual information becomes equal to $I_{\max}=\sum_{i=1}^{m}\sum_{j=1}^n p_{ij}p_{j}^{*}\log_2\frac{p_{ij}}{\sum_{k=1}^n  p_{ik}p_{k}^{*}}$. In case that ER is active, the average energy harvested is written as $\mathcal{E}_{\min}=\sum_{j=1}^{n}p_{j}^{*}I_0(Bh_E d(x_{j}))$.
\begin{remark}\label{rm1}
For the case where $d(A)\leq \overline{A}\approx 1.665$ (peak output amplitude  \cite{SHA}) and $A^2\leq  \sigma_x^2$, the input distribution is equiprobable binary i.e., $\Pi_A=\frac{1}{2}(\delta_{-A}+\delta_{+A})$; in this case, we have $I_{max}=1-H_2(P_e)$ and $\mathcal{E}_{\min}=I_0(Bh_E d(A))$, where $H_2(x)$ is the binary entropy with probability $x$, and $P_e=\int_{0}^{\infty} p(y|-A) dy$. In this case and according to Proposition 1, we can see that there is not a trade-off between information/energy and the same input distribution (i.e., equiprobable binary with mass points at $\pm A$) maximizes both information and energy transfer simultaneously.
\end{remark}    

Then, we consider the case where both ER and IR are active/present. The IE-CR is defined as 
\begin{align}
\mathcal{C}(\sigma_x^2,A)=\bigg\{(I,\mathcal{E}):\; &I\leq I_{\max},\; \mathcal{E}\leq \mathcal{E}_{\max},\; \nonumber  \\
& \mathbb{E}[x^2]\leq \sigma_x^2,\; |X|\leq A \bigg\}. \label{region}
\end{align}
To characterize the boundary of the IE-CR, we observe that when $I\leq I_{min}$, the maximum average harvested energy is given by the input distribution that achieves the rate tuple $(I_{\min},\mathcal{E}_{\max})$, given by the solution to (P1). On the other hand, when $\mathcal{E}\leq \mathcal{E}_{\min}$, the maximum information rate is given by the input distribution that achieves the rate tuple $(I_{\max},\mathcal{E}_{\min})$, given by the solution to (P3). The other points of the boundary $I_{\min} \leq I \leq I_{\max}$ and $\mathcal{E}_{\min} \leq \mathcal{E} \leq \mathcal{E}_{\max}$ can be found by solving a new optimization problem, which is similar to (P3) with the extra constraint $\mathcal{E}_{\min}\leq \mathbb{E}[I_0(Bh_E|\hat{X}|)]\leq \mathcal{E}_{\max}$. Since the extra constraint is linear over the input distribution $F_X$, the optimization problem is still convex and can  solved by using standard convex optimization tools e.g. CVX.   

\begin{figure}[t]
\centering
\includegraphics[width=0.96\linewidth]{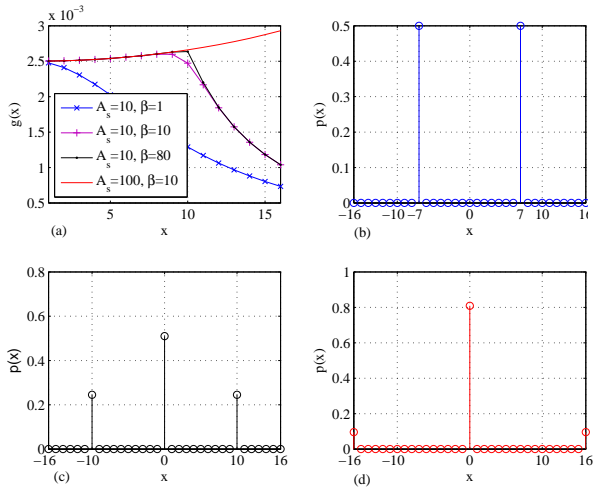}
\vspace{-0.3cm}
\caption{(a) The function $g(x)$ for different parameters of the SSPA model; we also assume $A=16$, $B=0.1$, and  $\sigma_x^2=49$, (b) Input distribution for $g(x)\searrow (\sigma_x, A)$ , (c) Input distribution for  $g(x)\nearrow (\sigma_x, A_s)$ and $g(x)\searrow (A_s, A)$, and (d) Input distribution for $g(x)\nearrow (\sigma_x, A)$.}\label{fig1}
\end{figure}

\subsection{Digital predistortion}
In this section, we study the IE-CR for the case where a PD is applied to the input signal before HPA. The purpose of PD is to compansate the non-linear HPA effects and linearize the non-saturation regime of HPA. In case that HPA function $d(r)$ is deterministic and known at the transmitter, an ideal PD corresponds to the function $q(r)$ i.e.,   
\begin{align}
q(r)=\left\{\begin{array}{l} A_s,\;\;\;\;\;\;\;\;\;\;\;\;\;\;\;\;\;\;\;\;\;\;\;\;\;\;\;\;\;\;\;\;\;\text{If}\;r\geq A_s,  \\
d^{-1}(r)=\frac{r}{\left[1-\left(\frac{r}{A_s}\right)^{2\beta} \right]^{\frac{1}{2\beta}}},\;\text{If}\; -A_s<r<A_s, \\
-A_s,\;\;\;\;\;\;\;\;\;\;\;\;\;\;\;\;\;\;\;\;\;\;\;\;\;\;\;\;\;\;\text{If}\;r\leq -A_s. \end{array} \right.
\end{align}  
By using similar analytical steps with the HPA case (i.e., solving optimization problems (P1), (P2) and (P3)), the information energy capacity region is expressed by \eqref{region} with two basic modifications i.e.,  i) the AP constraint is replaced by  $\mathbb{E}[q(x)^2]\leq \sigma_x^2$, and ii) HPA's output is equal to $\hat{X}=d(q(X))$. These two modifications do not affect the nature and the characteristics of the problem (discreteness of the input distribution, convexity over $\pmb{p}$ etc.) and therefore the proposed mathematical framework can be applied accordingly. It is worth noting that $r\geq d(r)$ and therefore PD penalizes the AP constraint (increases transmit power), while it facilitates the objective functions in (P1)-(P3).

\section{Numerical results}

Computer simulations have been carried out to evaluate the impact of HPA in terms of IE-CR; for the sake of simplicity, we assume $h_I=h_E=1$ without loss of generality.

Fig. \ref{fig1} deals with the input mass distribution for different system configurations when IR is not active and the target is to maximize the average harvested energy. We assume $A=16$, $\sigma_x^2=49$ and thus $\sigma_x^2\leq A^2$; Fig. \ref{fig1}(a) plots the function $g(x)$ for the configurations considered. For the case where $A_s=10$, $\beta=1$ (see \ref{fig1}(b)), we have $g(x)\searrow (\sigma_x, A)$ and the optimal input mass distribution consists of two mass points at $\pm \sigma_x$. On the other hand, for the scenario where $A_s=10<A$ and $\beta=80$ (see \ref{fig1}(c)), the nonlinearity in the HPA transition region disappears and thus  $g(x)\nearrow (\sigma_x, A_s),\;g(x)\searrow (A_s, A)$; in this case the input distribution consists of three mass points at $\{\pm A_s, 0\}$. Finally, for the scenario where $A<A_s=100$ and $\beta=10$, we have $g(x)\nearrow (\sigma_x, A)$ and the optimal input distribution is $\{\pm A, 0\}$. The main observations of Fig. \ref{fig1} are inline with Proposition \ref{pr1}.      

\begin{figure}[t]
\centering
\includegraphics[width=\linewidth]{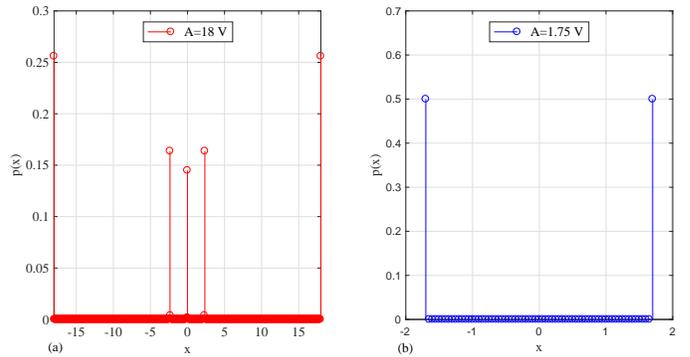}
\vspace{-0.3cm}
\caption{Input mass probability distribution for maximum information transfer;  $\beta=1$, $ B=0.5$, $\sigma_x^2=30$ dB, and (a) $A=18$, and (b) $A=1.75$.}\label{fig2}
\end{figure}

In Fig. \ref{fig2}, we show the input distribution for the case where ER is not present and the goal of the system is to maximize information Shannon capacity; the setting is $A_s=5$, $\beta=1$, $\sigma_x^2=30$ dB, and $B=0.5$. For the case where $A=18$ (Fig. \ref{fig2}.(a)), we can see that the optimal input distribution is discrete with a finite number of mass probability points. In Fig. \ref{fig2}.(b), we examine the special case of small $A$ i.e.,  $A=1.75$ (with $d(1.75)=1.6518< \overline{A}\approx 1.665$ \cite{SHA}) and as it can be seen the optimal input distribution is binary with two mass points at $\pm A$; this observation is inline with Remark \ref{rm1}. 

Fig. \ref{fig3} shows the fundamental information-energy capacity region for the considered SWIPT system with HPA nonlinearities at the transmitter. The simulation setup assumes $A_s=5$, $\beta=1$, $\sigma_x^2=30$ dB and $B=0.5$; the case without HPA degradation is used as a benchmark (no-HPA). The first remark is that HPA nonlinearities significantly reduce the achieved IE-CR in comparison to the no-HPA case; the negative effects of HPA are more critical as the PP constraint increases. Another important observation is that for low $A$ (i.e., $A=1.75$ with $d(A)<\overline{A}$), there is not a tradeoff between information and energy and thus the same input distribution maximizes both information and energy transfer (Remark \ref{rm1}). Finally, in the curve corresponding to $A=10$, we can see the key points of the boundary of the information-energy capacity region, which are defined in \eqref{region}.     

Finally, Fig. \ref{fig4} deals with the impact of PD on the IE-CR; we study configurations with a different parameter $\beta$. As we can see, the application of a PD on the input signal, limits the negative effects of HPA and enlarges the IE-CR. However, the observed gain decreases as the smoothness parameter $\beta$ increases, since the associated power for the inversion $d^{-1}(r)$ increases; for $\beta=10$, the transition region is almost linear and the application of PD slighly decreases the IE-CR performance.   

\begin{figure}[t]
\centering
\includegraphics[width=\linewidth]{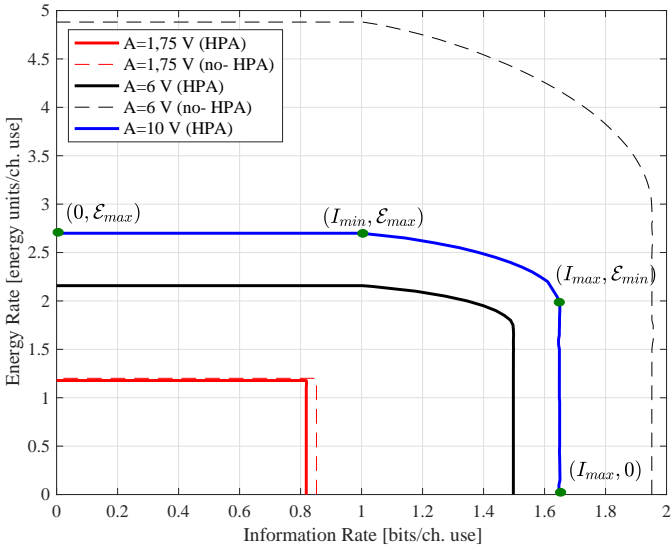}
\vspace{-0.3cm}
\caption{Information-energy capacity region; $A_{s}=5$, $\beta=1$, $B=0.5$, $\sigma_x^2=30$ dB.} \label{fig3}
\end{figure}

\begin{figure}[t]
\centering
\includegraphics[width=\linewidth]{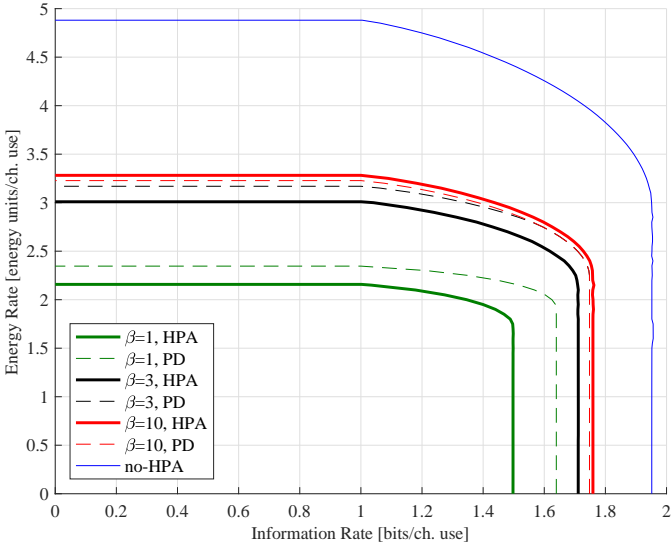}
\vspace{-0.3cm}
\caption{Information-energy capacity region for the predistortion scheme; $A=6$, $A_{s}=5$, $B=0.5$, $\sigma_x^2=30$ dB.} \label{fig4}
\end{figure}

\appendix

We consider the functions $g_1(x)=I_0(\theta x)$ and $g_2(x)=g_1(d(x))$ where $\theta$ is a constant. The function $d(x)$ is monotonically increasing function i.e., the first derivative equals to $d'(x)=1/[1+(x/A_s)^{2\beta}]^{\frac{1}{2\beta}+1}>0$, $\forall\ x$. Given that $g_1(x)$ is monotonically increasing function for $x>0$, the composite function $g_2(x)$ is an increasing function for $x>0$ (composition of two increasing functions). 

For the case where $A^2 \leq \sigma_x^2$, the PP constraint dominates and due to the monotonicity and even symmetry of $g_2(x)$, the optimal distribution consists of two mass points at $-A$ and $A$ with probabilities $p_1$ and $p_2=1-p_1$, respectively. Therefore, the average harvested energy becomes equal to $\mathcal{E}_{\max}=g_2(A)$. Although $p_1$ can take any value in $(0, 1)$ without affecting the maximum average harvested energy, we assume $p_1=1/2$ to maximize the information transfer (in case that IR is active). 

When $\sigma_x^2<A^2$, we examine also the case where the mass points are located at the region $(\sigma_x, A)$. Similarly to the previous case (i.e., $A^2\leq \sigma_x^2$), let $x_0=\sigma_x$ the point of increase of a distribution $F_X$ with probability $1$; we construct a new distribution $F'_X$ by removing $x_0$ and adding two mass points at the locations $0$ and $y \in (\sigma_x, A)$ with probabilities $1-\sigma_x^2/y^2$ and $\sigma_x^2/y^2$, respectively. We can show that this transformation decreases/increases the harvested energy depending on the monotonicity of the function $g(x)=(g_2(x)-1)/x^2$. More specifically, if $g(x)$ is a decreasing function in $(\sigma_x, A)$  i.e.,   
\begin{align}
& \frac{g_2(\sigma_x)-1}{\sigma_x^2}>\frac{g_2(y)-1}{y^2} \nonumber \\
& \Rightarrow g_2(\sigma_x)>\left(1-\frac{\sigma_x^2}{y^2} \right)g_2(0)+\frac{\sigma_x^2}{y^2}g_2(y), \label{con1}
\end{align}     
with $g_2(0)=1$, $F_X'$ decreases the average harvested energy and thus $F_X$ is the optimal input distribution; by following similar arguments as before, the optimal distribution consists of two points at $-\sigma_x$ and $\sigma_x$ with probabilities $1/2$, and the maximum harvested energy becomes equal to $\mathcal{E}_{\max}=g_2(\sigma_x)$. On the other hand, if $g(x)$ is an increasing function in $(\sigma_x, A)$, the inequality in \eqref{con1} holds with the reverse direction and $y=A$ maximizes the average harvested energy. In this case, the optimal mass function consists of three points at the locations $-A$, $A$ and $0$ with probabilities $p_{1}=p_{2}=\frac{\sigma_x^2}{2A^2}$ and $p_{0}=1-\frac{\sigma_x^2}{A^2}$, respectively. Finally, in case that $A_0 \in (\sigma_x, A)$, the function $g(x)$ is increasing in the interval $(\sigma_x, A')$ and decreasing in the interval $(A', A)$ and therefore we have $y=A'$; we note  $A'\approx A_s$ for $\beta>>1$. Equivalently, the optimal input distribution consists of three mass points at the locations $-A'$, $A'$ and $0$ with probabilities $p_{1}=p_{2}=\frac{\sigma_x^2}{2A'^2}$ and $p_{0}=1-\frac{\sigma_x^2}{A'^2}$.  For these two subcases (with three mass points), the maximum average energy is equal to $\mathcal{E}_{\max}\approx 2p_1 g_2(\mu)+p_0$ with $\mu=A$ and $\mu=A'$, respectively.

\section*{Acknowledgement}

This work has received funding from the European Research Council (ERC) under the European Union's Horizon 2020 research and innovation programme (Grant agreement No. 819819).

\end{document}